\pdfoutput=1
\documentclass[a4paper,12pt]{article}
\usepackage{amsmath,amsfonts,amsthm,amscd,amssymb,slashed,graphicx,tikz,stackengine}
\usepackage[colorlinks,allcolors=blue]{hyperref}
\usetikzlibrary{shapes,arrows}
\usetikzlibrary{matrix,arrows}
\usepackage{geometry}
\newcommand{\R}{\mathbb{R}}

\newcommand{\C}{\mathbb{C}}

\DeclareMathOperator{\SU}{SU}
\DeclareMathOperator{\U}{U}
\DeclareMathOperator{\SE}{SE}
\DeclareMathOperator{\E}{E}
\newcommand{\lf}[1]{\sigma^{#1}}

\newcommand{\p}{\partial}
\newcommand{\beq}{\begin{eqnarray}}
\newcommand{\eeq}{\end{eqnarray}}
\renewcommand{\i}{\mathrm{i}}
\renewcommand{\d}{\mathop{}\!\mathrm{d}}
\newcommand{\non}{\nonumber\\}
\newcommand{\commawedge}{\raisebox{-2pt}{\,\stackanchor[1pt]{$\wedge$}{,}}\,}

\newcommand{\HH}[1]{\mathbb{H}^{1}_{#1}}

\newtheorem{theorem}{Theorem}%
\newtheorem{proposition}[theorem]{Proposition}% 
\newtheorem{definition}{Definition}%
\newtheorem{lemma}{Lemma}
\newtheorem{corollary}{Corollary}

\begin{document}

\begin{titlepage}
\renewcommand{\thefootnote}{\fnsymbol{footnote}}
\begin{center}
{\huge Cartan connections for an infinite family of integrable vortices}\\

\bigskip
{\large Sven Bjarke Gudnason$^{1,2}$\footnote{gudnason(at)henu.edu.cn},
Calum Ross$^{3,4}$\footnote{Calum.Ross(at)edgehill.ac.uk}}\\

\bigskip
\bigskip
\noindent
$^1$Institute of Contemporary Mathematics, School of
  Mathematics and Statistics, Henan University, Kaifeng, Henan 475004,
  P.~R.~China\\
$^2$Department of Physics, Chemistry and Pharmacy,
  University of Southern Denmark, Campusvej 55, 5230 Odense M,
  Denmark\\
$^3$Department of Computer Science, Edge Hill University, St Helens Rd., Ormskirk L39 4QP, UK\\
$^4$Research and Education Center for Natural Sciences, Keio
  University, 4-1-1 Hiyoshi, Yokohama, Kanagawa 223-8521, Japan\\
  \bigskip
\begin{abstract}
An infinite family of integrable vortex equations is studied and related to the Cartan geometry of the underlying Riemann surfaces. This Cartan picture gives an interpretation of the vortex equations as the flatness of a non-Abelian connection. 
Solutions of the vortex equations also give rise to magnetic zero-modes for a certain Dirac operator on the lifted geometry.
The family of integrable vortex equations is parametrised by a positive number $n$, that is equal to unity in the standard case and an integer in the case of polynomial vortex equations; finally, it may be extended to any positive real number.
\end{abstract}
\end{center}
\end{titlepage}
\renewcommand{\thefootnote}{\arabic{footnote}}

\section{Introduction}

Manton put forward the idea of studying integrable vortices from the
point of view of the \emph{differential equations}, instead of the
Lagrangians, which gave rise to the notion of \emph{five} integrable
vortices \cite{Manton2} -- four of them were known in the literature
and one was new and is now dubbed the Bradlow vortex equation, due to its resemblance to the Bradlow limit of vortices on a compact domain or manifold \cite{Bradlow:1990ir}.
The Abelian integrable vortex equations describe a pair $(A,\phi)$ of a gauge potential and a complex scalar field, and can be put in the form
\begin{align}
  \d_A\phi\wedge \d z = 0,\label{eq:BPS1}\\
  \d A = \left(-C_0 + C|\phi|^2\right)\omega_0.\label{eq:BPS2}
\end{align}
The first equation is the usual self-duality equation giving rise
to a covariant holomorphic condition on the complex scalar field
$\phi\ :\ M_0\to\mathbb{C}$ on a manifold $(M_0,g_0)$ with metric
$g_0 = \Omega_0\d z\d\bar{z}$ and K\"ahler form $\omega_0$.
$\Omega_0$ is the conformal factor of the metric $g_0$ on $M_0$,
$\d A$ is the Abelian field strength,
$C_0$ is a constant that is fixed by the curvature of
$M_0$ by integrability and finally $C$ fixes the curvature of the
Baptista manifold that has the (degenerate) metric
$g = |\phi|^2g_0$,
valid away from the vortex zeros.
Indeed the geometry of $(M_0,g_0)$ is a fixed-curvature manifold with
curvature $K_0=C_0$ with $C_0=-1,0,1$ corresponding to the hyperbolic
plane ($\mathbb{H}^2$), the Euclidean plane $\E_2$ and the 2-sphere
($S^2$), respectively.

The general form for the vortex solutions to Manton's five vortex
equations is
\beq
|\phi|^2 = \frac{\left(1+C_0|z|^2\right)^2}{\big(1+C|f(z)|^2\big)^2}\left|\frac{\d f}{\d z}\right|^2,
\label{eq:phi^2solns}
\eeq
with the constant coefficient $C_0$ and $C$ corresponding to the type
of the vortex equations: $(C_0,C)=(-1,-1)$ is Taubes' equation,
$(C_0,C)=(0,1)$ is the Jackiw-Pi equation, $(C_0,C)=(1,1)$ is Popov's
equation, $(C_0,C)=(-1,0)$ is the Bradlow vortex equation and finally
$(C_0,C)=(-1,1)$ is the Ambj\o rn-Olesen vortex equation.
The vortex positions or vortex zeros correspond to the ramification points of the holomorphic function $f(z)$, i.e.~where its $z$-derivative vanishes.
The Bradlow vortex equation has a constant
vortex polynomial making the magnetic flux proportional to the area of
$M_0$, which in turn is related to the number of vortices on the
manifold. 
In this case, due to the simplicity of the equations, Gudnason and
Nitta found additional solutions \cite{Gudnason:2017jsn}, in addition to those of the form of
Eq.~\eqref{eq:phi^2solns}. 

Shortly after Manton's paper, Contatto and Dunajski showed that the
five vortex equations arise as symmetry reductions of anti-self-dual
Yang-Mills theory in four dimensions with different symmetry groups
\cite{CD}.
In particular, the symmetry groups for the reductions of 4-dimensional
Yang-Mills theory corresponding to $C_0=-1,0,1$ are $\SU(1,1)$,
$\SE_2$ and $\SU(2)$, respectively. 

Then in a series of papers by Ross and Schroers
\cite{Ross:2021,RS1,RS2}, the vortex equations which are characterised
by the reduction groups $\SU(1,1)$, $\SE_2$ and $\SU(2)$ are related
to the Maurer-Cartan structure on these group manifolds; the key idea
being that these group manifolds are circle bundles over Riemann
surfaces with a Hopf projection ($\pi$) from the group down to the
Riemann surface.

In Ref.~\cite{Gudnason:2022}, Gudnason proposed to extend the vortex
polynomial from a linear polynomial in $|\phi|^2$ to a quadratic
polynomial, which allowed for the inclusion of the BPS Chern-Simons
equation in the class of vortex equations.
The BPS Chern-Simons vortex equation is not integrable. Nevertheless,
four new integrable equations where identified in
Ref.~\cite{Gudnason:2022}\footnote{The most trivial new vortex equation (in addition to Manton's five equations) is the Laplace vortex equation found in Ref.~\cite{CD}, that simply consists of an empty vortex polynomial
($C_0=C=0$) with delta-function sources at the vortex centres; this was included in Ref.~\cite{Gudnason:2022} as well.
}
being the four known vortex equations
with $|\phi|^2$ replaced by $|\phi|^4$: These are the Jackiw-Pi, the
Taubes, the Popov and the Ambj\o rn-Olesen vortex equations.
Finally, it was mentioned in the discussion of
Ref.~\cite{Gudnason:2022} that these four vortex equations form an
infinite family of integrable vortex equations with $|\phi|^2$
replaced by $|\phi|^{2n}$, $n=2,3,\ldots$, as
\beq
\d A = \left(-C_0 + C_{2n}|\phi|^{2n}\right)\omega_0,
\label{eq:vtx2_2n_Gudnason:2022}
\eeq
with $n$ a positive integer and the constants taking the same values as their counterparts from Manton's equations.
That is, the vortex coefficients when non-vanishing have a normalised unit modulus ($|C|=1$), which as argued in Ref.~\cite{Manton2} can always be achieved by rescaling the scalar field $\phi$ and rescaling the spatial coordinate $z$.

Other developments that followed soon after Manton's paper on the five vortex equation \cite{Manton2} can be summarised as follows.
Bobenko and Sridhar established a connection between the five vortex equations and structure-preserving discrete conformal maps \cite{Bobenko:2017zso}.
Finally, Gudnason and Ross generalised the five vortex equations to include defects \cite{Gudnason:2021bkw}.
Further work emanating from the five vortex equations paper, which involves only a few or a single of the equations will not be listed here.

In this paper, we demonstrate that the possibility mentioned in Ref.~\cite{Gudnason:2022} of an
infinite family of integrable vortex equations is true and that
they have a geometric realisation in terms of Maurer-Cartan structure on
the group manifolds $\SU(1,1)$, $\SE_2$ and $\SU(2)$.
The treatment can be unified in two different ways: In one choice of
coordinates, the vortex equations are rescaled differently for each
$n=1,2,3,\ldots$, but are described by the same Cartan geometry with
the same Hopf-type projection to a Riemann surface;
in the other the vortex equations are naturally normalised, but the
geometry depends on $n$, which can be interpreted as the ``size'' 
or ``scale-factor'' of the Riemann surface.

The vortex equations will here be called the $n$-vortex equations, but
note that this should not be confused with the vortex winding number,
which is $N$ and independent from $n$.
Indeed, the vortex winding number, $N$, is given by 
\beq
N = \frac{1}{2\pi}\int_{M_0}\d A.
\eeq
We note that from the point of view of the equations, we can
generalise everything, including the geometries, to
$n\in\mathbb{R}_{>0}$, which does not correspond to vortex
polynomials, but is still well defined from the point of view of the equations.
We note, however, that in all cases we must have $n>0$ strictly
positive.

This paper is organised as follows.
In Sec.~\ref{sec:fixed_geometry} the geometry is held fixed corresponding to an $n$-dependent family of vortex equations, where the local geometry is presented in Sec.~\ref{sec:local_geometry}, the lift to higher-dimensional Cartan geometry is given in Sec.~\ref{sec:cartan_geometry} and finally, the zero-modes are discussed in Sec.~\ref{sec:zeromodes}.
In Sec.~\ref{sec:normalised_vtxeqs} the vortex equations are normalised as in Refs.~\cite{Manton2,Gudnason:2022}, but the geometry inherits non-trivial $n$ dependence that for the 2-sphere can be viewed as a radius proportional to $\sqrt{n}$.
Finally, Sec.~\ref{sec:conclusion} concludes with a brief discussion of how this work relates to some other ongoing work and mentions some future directions.

\section{Fixed geometry with infinite family of vortex equations}\label{sec:fixed_geometry}

\subsection{Local geometry on the Riemann surface}\label{sec:local_geometry}

We will start with the geometric picture, expanded from that given in
Ref.~\cite{Gudnason:2022}, i.e.~that $(M_0,g_0)$ is a Riemann
surface of Gauss curvature $K_0=C_0$ with metric
\beq
  g_0 = \frac{4}{\left(1+C_0|z|^{2}\right)^{2}}\d z\d\bar{z}
  = \Omega_0\d z\d\bar{z}
  = e_0\bar{e}_0,
\eeq
where 
\beq
e_0=\frac{2}{1+C_0|z|^{2}}\d z,
\label{eq:local_frame}
\eeq
is the local complex co-frame.
This complex co-frame obeys the structure equation 
\beq
\d e_0-\i\Gamma_0\wedge e_0=0,
\eeq
with the spin connection 
\beq
\Gamma_0=\i C_0\frac{z\d\bar{z}-\bar{z}\d z}{1+C_0|z|^{2}}.
\label{eq:spin_conn}
\eeq
The spin connection is also related to the complex co-frame through the
Gauss equation 
\beq
\mathcal{R}_0=\d\Gamma_0=\frac{\i}{2}K_0e_0\wedge\bar{e}_0,
\eeq
from which we can see that the scalar curvature is $K_0=C_0$ as promised.
For a holomorphic function $f:(M_0,g_0)\to (M_{2n},g_{2n})$ between two
such Riemann surfaces we can pull back the structure and Gauss
equations from $M_{2n}$ to $M_0$ and show that they imply the vortex
equations.
Define $f^{*}e_{2n}=\phi^{n}e_0$ and compute the pull back of the
structure equation as
\begin{align}
0 &=\d f^{*}e_{2n}-\i f^{*}\Gamma_{2n}\wedge f^{*}e_{2n}\non
 	&=\phi^{n}\left(\d e_0-\i f^{*}\Gamma_{2n}\wedge e_0\right)+n\phi^{n-1}\d\phi\wedge e_0\non
 	&=\phi^{n}\left(\d e_0-\i\Gamma_0\wedge
e_0\right)+n\phi^{n-1}\left(\d\phi-\i\phi A\right)\wedge e_0,
\label{eq:structure_e2n}
\end{align}
where we have defined the Abelian connection
\beq
A=\frac{\left(f^{*}\Gamma_{2n}-\Gamma_0\right)}{n}.
\label{eq:A}
\eeq
The first bracket in the last line of Eq.~\eqref{eq:structure_e2n}
vanishes by the structure equation on $M_0$ and the second bracket is
the self-dual vortex equation (notice that $e_0\propto\d z$).

The second vortex equation arises as the pull back of the Gauss
equation on $M_{2n}$
\beq
\d\Gamma_{2n}=C_{2n}\frac{\i}{2}e_{2n}\wedge\bar{e}_{2n}.
\eeq
to $M_0$:
\beq
\frac{1}{n}f^{*}\d\Gamma_{2n}&=\frac{\i C_{2n}}{2n}|\phi|^{2n}e_0\wedge\bar{e}_0,
\eeq
and add $\d\Gamma_0/n$ to obtain the Abelian field strength
\beq
F_{A}=\d A
=\left(-\frac{C_0}{n}+\frac{C_{2n}}{n}|\phi|^{2n}\right)\frac{\i}{2}e_0\wedge\bar{e}_0,
\label{eq:vtx2}
\eeq
which indeed is the second vortex equation with scalar curvatures $K_0=C_0$ and
$K_{2n}=C_{2n}$. 
Notice the scaling of vortex equations by $1/n$: this comes
naturally in the fixed geometry case.
As a comment on our notation: we normalise $C_0$ and $C_{2n}$ to unit modulus when non-vanishing throughout the paper\footnote{Explicitly, $C_0,\in\{-1,0,1\}$ and $C_{2n}\in\{-1,0,1\}$ with the requirement that the right-hand side of the equation has a non-negative integral. Also, it is only when $C_{2n}\neq 0$ that we get a new type of vortex equation. When it vanishes the vortex equations are known; it is either the Bradlow ($C_0=-1$) \cite{Manton2} or the Laplace ($C_0=0$) \cite{CD} vortex equation. 
}
for this reason we put the $n$-dependence explicitly in the above equation.
Compared with the vortex equation \eqref{eq:vtx2_2n_Gudnason:2022} which does not have any $n$-dependence (normalised coefficients when non-vanishing), the introduction of the $n$-dependence here is made to be able to keep the geometry fixed -- at the price of having a non-normalised vortex equation. 
Later in Sec.~\ref{sec:normalised_vtxeqs}, we will discuss the normalised vortex equation's case that leads to $n$-dependent geometries.

The vortex equations are equivalent to the structure and Gauss
equations for the degenerate form $\tilde{e}=f^*e_{2n}=\phi^{n}e_0$ on
$M_0$. 
This degenerate form ($\tilde{e}$) will possess 
zeros at the
points corresponding to the zeros of the vortices, which is a usual
property of the Baptista metric:
\beq
\tilde{g}=\tilde{e}\bar{\tilde{e}}=|\phi|^{2n}e_0\bar{e}_0=|\phi|^{2n}g_0.
\eeq

Following Ref.~\cite{Ross:2021}, we know that the structure and Gauss
equations on $M_{2n}$ are equivalent to the flatness of the
non-Abelian connection
\beq
\hat{A}=-\Gamma_{2n}t_0+\frac\i2\left(e_{2n}t_{-}-\bar{e}_{2n}t_{+}\right),
\label{eq:Ahat}
\eeq
where $t_0,t_{\pm}=t_{1}\pm \i t_{2}$ are the generators of a Lie
group with commutation relations 
\beq
[t_0,t_{\pm}]=\mp\i t_{\pm}, \qquad [t_{+},t_{-}]=-2\i C_{2n}t_0.
\label{eq:commutators}    
\eeq
Depending on $K_{2n}=C_{2n}$ we will either have an $\mathfrak{su}(2)$
($C_{2n}=1$) or an $\mathfrak{su}(1,1)$ ($C_{2n}=-1$) structure. 
The pullback of $\hat{A}$ by $f$ 
\beq
f^{*}\hat{A}=-\left(nA+\Gamma_0\right)t_0+\frac\i2\left(\phi^{n}e_0t_{-}-\bar{\phi}^{n}\bar{e}_0t_{+}\right),
\eeq
is a flat non-Abelian connection encoding the structure and Gauss
equations for the degenerate co-frame $\tilde{e}$: In other words, it
encodes the vortex equations.
To see this explicitly we compute the curvature $F_{f^{*}\hat{A}}$:
\begin{align}
  F_{f^{*}\hat{A}}&= \d f^*\hat{A} + \frac12\Big[f^*\hat{A}\commawedge f^*\hat{A}\Big]\non
  &=-n\left[\d A-\left(-\frac{C_0}{n}+\frac{C_{2n}}{n}|\phi|^{2n}\right)\frac{\i}{2}e_0\wedge\bar{e}_0\right]t_0\\
&\phantom{=\ }+\frac{\i}{2}n\left[\phi^{n-1}\left(\d\phi -\i A\phi\right)\wedge e_0 t_{-}-\bar{\phi}^{n-1}\left(\d\bar{\phi} +\i A\bar{\phi}\right)\wedge \bar{e}_0 t_{+}\right],
\end{align}
and see that its vanishing is equivalent to the vortex equations being
satisfied.

By using the definition of the generalised Baptista metric, $\Omega_{2n}=|\phi|^{2n}\Omega_0$, the solution for the scalar field can readily be written down in the holomorphic gauge as
\begin{equation}
\phi^n = 
\frac{\left(1+C_0|z|^2\right)}{\left(1+C_{2n}|f(z)|^2\right)}\frac{\d f}{\d z},
\end{equation}
with $f(z)$ a holomorphic function whose ramification points correspond to vortex positions.
The gauge field is then by the self-dual equation \eqref{eq:BPS1}
\beq
A = -\frac{\i}{n}\p_{\bar{z}}\log\left(\phi^n\right)\d\bar{z}.
\eeq

\subsection{Cartan geometry with circle bundles over Riemann surfaces}\label{sec:cartan_geometry}

In this section, we generalise the results of
Refs.~\cite{Ross:2021,RS1,RS2} where the usual vortex equations on
$\R^{2},S^{2},\mathbb{H}^{2}$ are related to the Maurer-Cartan structure on the
group manifolds $\SE_{2}$, $\SU(2)$, or $\SU(1,1)$, respectively, to the
infinite family of $n$-vortices.

The key idea is that these group manifolds are circle bundles over the
Riemann surfaces.
We then have a Hopf type projection, $\pi$, from the group down to the
Riemann surface.
The group structure is encoded in the Lie algebra
\eqref{eq:commutators}. 
Using the notation $\mathbb{H}^{1}_{C}$ for the group manifolds
we define them as subsets of $\C^{2}$ as 
\beq
\mathbb{H}^{1}_{C}=\Big\{(z_{1},z_{2})\in \C^{2}\Big||z_{1}|^{2}+C|z_{2}|^{2}=1\Big\},
\eeq
with $\mathbb{H}^{1}_{1}=\SU(2)$, $\mathbb{H}^{1}_0=\SE_{2}$, and
$\mathbb{H}^{1}_{-1}=\SU(1,1)$.
A typical element of the group is
\beq
h=\begin{pmatrix}
z_{1}&-C\bar{z}_{2}\\
z_{2}&\bar{z}_{1}
\end{pmatrix}.
\eeq
The Maurer-Cartan structure is given by
\beq
h^{-1}\d h =\sigma^{0}t_0+\sigma^{1}t_{1}+\sigma^{2}t_{2}=\sigma^{0}t_0+\frac{1}{2}\left(\sigma t_{-}+\bar{\sigma}t_{+}\right),
\eeq
with the left-invariant one forms being 
\begin{align}
\sigma&=\sigma^{1}+i\sigma^{2}=2\i\left(z_{1}\d z_{2}-z_{2}\d z_{1}\right), \\
\sigma^{0}&=\i\left(\bar{z}_{1}\d z_{1}+C\bar{z}_{2}\d z_{2}-z_{1}\d\bar{z}_{1}-C z_{2}\d\bar{z}_{2}\right).
\end{align}
The projection is
\beq
\pi \ :\ \mathbb{H}^{1}_{C}\to M, \qquad
h\mapsto z=\frac{z_{2}}{z_{1}},
\eeq
and there is a (local) section
\beq
s\ :\ z\mapsto \frac{1}{\sqrt{1+C|z|^{2}}}\begin{pmatrix}
1 & -C\bar{z}\\
z & 1
\end{pmatrix}.
\eeq
The co-frames on the Riemann surfaces are then related to the
left-invariant one forms through 
\begin{align}
\pi^*e &= -\i\frac{\bar{z}_{1}}{z_{1}}\sigma,\qquad&
s^*\sigma &= \i e,\label{eq:pi_s1}\\
\pi^*\Gamma&=-\sigma^{0}+\i\d\ln\left(\frac{z_{1}}{\bar{z}_{1}}\right),\qquad&
s^*\sigma^0 &= -\Gamma.\label{eq:pi_s2}
\end{align}
For the integrable vortex equations, we thus have the following picture:
\begin{center}
  \begin{tikzpicture}[scale=2.0]
    \draw (0,1) node [left] {$\HH{C_0}$};
    \draw [->] (0,1) -- (1,1);
    \draw [->] (-0.2,0.8) -- (-0.2,0.2);
    \draw (0.5,1.1) node {\small $U$};
    \draw (0,0.5) node [left] {$\pi$};
    \draw (1,1) node [right] {$\HH{C_{2n}}$};
    \draw (0,0) node [left] {$M_0$};
    \draw [->] (0,0) -- (1,0);
    \draw [->] (1.2,0.8) -- (1.2,0.2);
    \draw (0.5,0.1) node {\small $f$};
    \draw (1.4,0.5) node [left] {$\pi$};
    \draw (1,0) node [right] {$M_{2n}$};
  \end{tikzpicture}
\end{center}
Here and in the above, $C$ can be either $C_0$ or $C_{2n}$ depending on whether the projection is used at the left-hand or right-hand side of the above diagram, and the same goes for the complex co-frame $e$ and the spin connection $\Gamma$.
$U$ is a bundle map covering the rational function $f$ which encodes a vortex. This diagram means that $\pi\circ H=f\circ \pi$. As in Ref.~\cite{Ross:2021}, the bundle map $U$ has the following form
\begin{equation}
    U:(z_{1},z_{2})\mapsto \frac{1}{\sqrt{\vert F_{1}\vert^{2}+C_{2n}\vert F_{2}\vert^{2}}}\begin{pmatrix}
        F_{1}&-C_{2n}\bar{F}_{2}\\
        F_{2}&\bar{F}_{1}
    \end{pmatrix},
\end{equation}
with $F_{i}:\HH{C_{0}}\to \C$ where $\vert F_{1}\vert^{2}>-C_{2n}\vert F_{2}\vert^{2}$. The functions $F_{i}$ are explicitly related to the rational function $f$ through $f=s^{*}\left(F_{2}/F_{1}\right)$. Sometimes it is convenient to split $f$ into two functions $f_{1},f_{2}$ as $f=f_{2}/f_{1}$. See Refs.~\cite{Manton2,Manton1, Witten1,Ross:2021} for explicit examples of the rational functions used to construct vortices.

The relation between $\hat{A}$ on $M_{2n}$ and $h^{-1}\d h$ on $\HH{C_{2n}}$ is
\beq
s^{*}\left(h^{-1}\d h\right)&=\hat{A},
\eeq
whereas moving up the diagram, we have by the identities
\eqref{eq:pi_s1}-\eqref{eq:pi_s2}, that
\beq
\pi^{*}\hat{A}&=r_{z_{1}}^{-1}\left(h^{-1}\d h\right)r_{z_{1}}+r_{z_{1}}^{-1}\d r_{z_{1}},
\eeq
with 
\begin{equation*}
    r_{P}=\begin{pmatrix}
        \frac{\overline{P}}{|P|}&0\\
        0& \frac{P}{|P|}
    \end{pmatrix},
\end{equation*}
the map from $\HH{C_0}\backslash \left(\bigcup S_{q_{j}}\right)\to \HH{2n}$, where $S_{q_{j}}$ are the circles where the function $P(z_{1},z_{2})$ vanishes. In this case $P(z_{1},z_{2})=z_{1}$ so we just have to remove the fibre above $z_{1}=0$. 

Then modifying Definition 4.2 in Ref.~\cite{Ross:2021}, we define an
$n$-vortex configuration on $\HH{C_0}$ as follows:
\begin{definition}
The pair $(\Phi,A)$ of a complex function and a one-form is called an
$n$-vortex configuration on $\HH{C_0}$, if it satisfies the $n$-vortex
equations
\begin{align}
    \left(\d \Phi +\i A \Phi\right)\wedge \sigma &=0,\\
    F_{A}=\d A&= -\frac{\i }{2}\left(\frac{C_0}{n}-\frac{C_{2n}}{n}|\Phi|^{2n}\right)\bar{\sigma}\wedge\sigma.
\end{align}
\end{definition}
As is true for the $n=1$ case, these equations have a $\U(1)$ gauge invariance,
\beq
    \left(\Phi,A\right)\mapsto \left(e^{-\i \beta}\Phi, A+\d \beta\right), \qquad \text{where } \beta \in\C^{\infty}\left(\HH{0}\right).
\eeq
Theorem 4.3 in Ref.~\cite{Ross:2021} then applies to these vortex
configurations with minimal modifications, namely the vortex gauge for
a flat $\text{Lie}\left(\HH{C_{2n}}\right)$ connection is now 
\beq
  \mathcal{A}=\left(nA+\lf{0}\right)t_0+\frac{1}{2}\Phi^{n}\sigma t_{-}+\frac{1}{2}\bar{\Phi}^{n}\bar{\sigma}t_{+},
\eeq
and the bundle map $U$ which covers the rational map $f$ satisfies
\beq
  \Phi^{n}\sigma=U^{*}\tau , \qquad nA=U^{*}\tau^{0}-\lf{0}.
\eeq
In Ref.~\cite{Ross:2021}, the notation $\tau^{a}$ was used for the left-invariant one forms on the target of $U$, here $\HH{C_{2n}}$.
Corollary 4.5 in Ref.~\cite{Ross:2021}, then
gives that 
\beq
f^{*}\hat{A}=r_{f_{1}}^{-1}s^{*}\mathcal{A}r_{f_{1}}+r_{f_{1}}^{-1}\d r_{f_{1}},
\eeq
with $f=\frac{f_{2}}{f_{1}}$, the rational function between the
Riemann surfaces and $r_{p}:M_0\backslash \{q_{j}\} \to \HH{C_{2n}}$.

\subsection{Zero-modes}\label{sec:zeromodes}

Following the result presented in Ref.~\cite{Ross:2021} for Manton's
five vortex equations, the $n$-vortex equations considered here can
also be described as zero-modes of twisted Dirac operators on the group
manifold. In this fixed geometry case, the geometry of the group manifold is unchanged from that considered in Ref.~\cite{Ross:2021}, so the only difference is that some results need to include explicit factors of $n$ to account for the factors of $n$ in the $n$-vortex equations. To avoid the singularity in the metric corresponding to a degenerate
Killing form, in this section we shall restrict to the case of
$C_0\neq 0$. 

The Killing form on $\HH{C_0}$ gives rise to the metric
\beq
\d s^{2}_{C_0}=\frac{1}{C_0}\left(\d x^{0}\right)^{2}+\left(\d x^{1}\right)^{2}+\left(\d x^{2}\right)^{2},
\label{eq:Killing_metric}
\eeq
where we make use of the oriented (pseudo) orthonormal co-frame
$(\d x^{0},\d x^{1},\d x^{2})$ with volume form
$\d\text{Vol}=\d x^{0}\wedge\d x^{1}\wedge\d x^{2}$.
We call the Lie algebra equipped with this metric
\beq
    \R^{3}_{C_0}=\left(\R^{3},\d s^{2}_{C_0}\right),
\eeq
and once we have zero-modes on the group manifold we can pull them back
to the Lie algebra to construct zero-modes there.
This requires using the inverse stereographic and gnomonic maps
constructed in Ref.~\cite{Ross:2021}: 
\beq
H,G : \R^{3}_{C_0}\to \HH{C_0},
\eeq
which are maps on the subspace
\beq
\mathcal{I}=\left\{\left.(x^{0},x^{1},x^{2})\in \R^{3}_{C_0}\right|C_0r^{2}>-1 \right\},
\eeq
where the distance $r$ is defined with respect to the
metric \eqref{eq:Killing_metric} as:
\beq
r^{2}=C_0(x^{0})^{2}+(x^{1})^{2}+(x^{2})^{2}.
\eeq
For $C_0=1$, corresponding to $\HH{1}=SU(2)$, $r^{2}>-1$ always holds
true, whereas for $C_0=-1$:
\beq
r^{2}=-(x^{0})^{2}+(x^{1})^{2}+(x^{2})^{2}<1,
\eeq
which means that $\mathcal{I}$ is the interior of a single-sheeted
hyperbola.
Adapting the expressions from Ref.~\cite{Ross:2021} we have 
\begin{align}
    H(\vec{x})&=\frac{1}{1+C_0r^{2}}\begin{pmatrix}
        1-C_0r^{2}+2\i x^{0}&2\i C_0(x^{1}-\i x^{2})\\
        2\i (x^{1}+\i x^{2})& 1-C_0r^{2}-2\i x^{0}
    \end{pmatrix},\\
    G(\vec{x})&=\frac{\mathbb{I}-2\vec{x}\cdot\vec{t}}{\sqrt{1+C_0r^{2}}}.
\end{align}
These matrices are related through $H(\vec{x})=(G(\vec{x}))^{2}$ and
are the key to relating the Dirac operator on $\HH{C_0}$ to the Dirac
operator on $\R^{3}_{C_0}$.
\begin{lemma}
    If $\Psi:\HH{C_0}\to \C^{2}$ is a magnetic Dirac mode of the Dirac operator $D_{\HH{C_0}, A}$, for the $\U(1)$ gauge field $A$, then
    \beq
        \Psi_{H}=G\Omega^{-1}H^{*}\Psi,
    \eeq
    is a magnetic Dirac mode on $\R^{3}_{C_0}$ with the Dirac operator $D_{\R^{3}_{C_0},H^{*}A}$.
\end{lemma}

The proof is given as the proof of Lemma 5.3 in Ref.~\cite{Ross:2021},
as the $n$-dependence is absent in the construction of the Dirac
zero-modes there is no difference.

\begin{definition}
  A vortex magnetic mode is a pair $(\Psi,A)$ of a spinor and a one
  form $A$ on $\HH{C_0}$ which satisfy the equations 
    \beq
        D_{\HH{C_0}, A}\Psi =0, \qquad F_{A}=-\frac{C_{2n}}{C_0n}4\i |\Psi|^{2n-2}\star\Psi^{\dagger}h^{-1}\d h\Psi+\frac{C_0}{4n}\lf{1}\wedge\lf{2}.
    \eeq
\end{definition}
\begin{proposition}
    Given a vortex configuration $(\Phi,A)$ on $\HH{C_0}$, the pair
    \beq
        \Psi=\begin{pmatrix}
            \Phi\\
            0
        \end{pmatrix}, \qquad A'=A+\frac{3}{4}\lf{0},
    \eeq
    is a vortex magnetic mode on $\HH{C_0}$.
\end{proposition}
\begin{proof}
  Since $(\Phi,A)$ is a vortex configuration it solves
  \beq
  X_0\Phi+\i A_0\Phi=0, \qquad X_{+}\Phi+\i A_{+}\Phi = 0,
  \eeq
  where $X_a$ are dual left-invariant vector fields that generate the right action $h\to h t_a$ with commutation relations
  \beq
  [X_0,X_\pm]=\mp\i X_\pm, \qquad
  [X_+,X_-]=-2\i C X_0,
  \eeq
  and in terms of complex coordinates can be written as
  \begin{align}
  X_- &= -\i\left(\bar{z}_1\p_2 - C\bar{z}_2\p_1\right),\qquad X_+=\overline{X}_-,\\
  X_0 &= -\frac{\i}{2}\left(z_1\p_1 + z_2\p_2 - \bar{z}_1\bar{\p}_1 - \bar{z}_2\bar{\p}_2\right),
  \end{align}
  yielding the (only) non-vanishing combinations
  \beq
  \sigma^0(X_0) = 1, \qquad
  \sigma(X_-) = \bar{\sigma}(X_+) = 2.
  \eeq
  While for $\Psi$ to be a magnetic Dirac mode it needs to satisfy 
  \beq
  D_{\HH{C_0}, A}\Psi =0,
  \eeq
  which is equivalent to
  \begin{align*}
    X_0\Phi+\i A_0'\Phi-\frac{3\i }{4}\Phi&=0,\\ 
    X_{+}\Phi+\i A_{+}\Phi&=0.
  \end{align*}
  Using $A'_0=A'(X_0)=A_0+\frac{3}{4}$ and $A'_{+}=A_{+}$, we see that
  the vortex configuration equations imply that $\Psi$ is a magnetic
  Dirac mode.

  Now we turn to the second equation.
  Note that a spinor with the form of $\Psi$ satisfies
  \begin{align}
    -\frac{C_{2n}}{C_0n}4\i |\Psi|^{2n-2}\star\Psi^{\dagger}h^{-1}\d h\Psi
    &=\frac{C_{2n}}{n}|\Phi|^{2n}\lf{2}\wedge\lf{1},
  \end{align}
  while
  \begin{align}
    F_{A'}  &=F_{A}+\frac{3}{4}C_0\lf{2}\wedge\lf{1}=\left(\frac{C_0}{n}-\frac{C_{2n}}{n}|\Phi|^{2n}\right)\lf{1}\wedge\lf{2}+\frac{3}{4}C_0\lf{2}\wedge\lf{1}\\
    &=-\frac{C_{2n}}{C_0n}4\i |\Psi|^{2n-2}\star\Psi^{\dagger}h^{-1}\d h\Psi+\frac{C_0}{4n}\lf{1}\wedge\lf{2},
  \end{align}
  which is the non-linear equation in the definition of a vortex
  magnetic mode. 
\end{proof}

\begin{corollary}
  Given a holomorphic map $f:M_0\to M_{2n}$ which determines a vortex,
  this gives rise to a smooth vortex magnetic mode on
  $\mathcal{I}\subset \R^{3}_{C_0}$.
  Lifting the vortex to a vortex configuration $(\Phi,A)$, the vortex
  magnetic mode is 
    \beq
    \Psi_{H}=\Omega^{-1}G\begin{pmatrix}
    H^{*}\Phi\\
    0
    \end{pmatrix},
    \qquad A'=H^{*}\left(A+\frac{3}{4}\lf{0}\right).
    \eeq
\end{corollary}
The proof of this result follows by stringing together all of the
other results.

\section{Normalised vortex equations with infinite family of
  geometries}\label{sec:normalised_vtxeqs}

The second BPS equation for the Taubes vortex with physical constants included reads
\beq
\d A = e^2\left(v^2 - |\psi|^2\right)\omega_0,
\eeq
which by the field redefinition $\phi=\psi/v$ (where $v>0$) can be written as
\beq
\d A = m^2\left(1 - |\phi|^2\right)\omega_0,
\eeq
and $m^2$ can be stowed away by changing to dimensionless units $z\to \frac{1}{m}z$ and $A\to m A$ -- this takes us to the vortex equation \eqref{eq:BPS2}, which in the Taubes case corresponds to $C_0=C=-1$.
Indeed, the parameter $m$ is physically the mass of the photon as well as the mass of the scalar field -- they are equal in the BPS case -- and is given by $m=ev$ with $e$ the gauge coupling constant and $v$ the vacuum expectation value (VEV) of the Higgs potential for the scalar field.

We will continue to work in dimensionless units where $e v$ have been absorbed into the length units of the coordinates on $M_0$, but the difference in this section compared with the approach in Sec.~\ref{sec:fixed_geometry} is that here we will also include the factor of $1/\sqrt{n}$ in the mass when rescaling the coordinates on $M_0$.
In order to avoid confusion, we define the new coordinate
\beq
w = \frac{1}{\sqrt{n}} z,\qquad
\d w = \frac{1}{\sqrt{n}} \d z,
\label{eq:vars_zw}
\eeq
such that the $n$-vortex equation \eqref{eq:vtx2} in $z$ coordinates is just Eq.~\eqref{eq:vtx2}:
\beq
F_{A}=\d A
=\left(-\frac{C_0}{n}+\frac{C_{2n}}{n}|\phi|^{2n}\right)\omega_0,
\label{eq:vtx2_repeat}
\eeq
but in $w$ coordinates it becomes Eq.~\eqref{eq:vtx2_2n_Gudnason:2022}:
\beq
F_{A}=\d A
=\left(-C_0+C_{2n}|\phi|^{2n}\right)\omega_0,
\label{eq:vtx2_2n_Gudnason:2022_repeat}
\eeq
as defined in Ref.~\cite{Gudnason:2022} with the standard normalisation (we have repeated the two equations here for convenience).
Notice that the discussion of the change of coordinates is only true locally: a consistent construction requires a change in the geometry.
That is, the local co-frame giving rise to the second vortex equation \eqref{eq:vtx2} (and \eqref{eq:vtx2_repeat}) is Eq.~\eqref{eq:local_frame}, whereas if we change the local co-frame to\footnote{This is not the same as performing the change of variable using Eq.~\eqref{eq:vars_zw}, since there is no factor of $\sqrt{n}$ in the denominator of the local co-frame field \eqref{eq:local_frame_new}.}
\beq
e_0 = \frac{2\d w}{1 + n C_0 |w|^2},
\label{eq:local_frame_new}
\eeq
the spin connection is
\beq
\Gamma_0 = \i n C_0 \frac{w\d\bar{w} - \bar{w}\d w}{1+n C_0|w|^2},
\eeq
the scalar curvature becomes $K_0=nC_0$ which follows from Gauss' equation
and the second vortex equation becomes that of Eq.~\eqref{eq:vtx2_2n_Gudnason:2022} (and \eqref{eq:vtx2_2n_Gudnason:2022_repeat}) since the overall factor of $n$ in $\Gamma_0$ and in $f^*\Gamma_{2n}$ is cancelled by the $1/n$ in the definition of the Abelian connection \eqref{eq:A}.
Changing back to the $z$ coordinates, the new local co-frame field becomes
\beq
e_0 = \frac{2\sqrt{n} \d z}{1 + C_0 |z|^2},
\eeq
which can be interpreted, for $C_0=1$, as the geometry of a 2-sphere of radius $R=\sqrt{n}$. 
With this interpretation, it is also clear why $n>0$ should be strictly positive.
The $n$-dependent radius of $g_0=e_0\bar{e}_0$ carries over to the Baptista metric $\tilde{g}=|\phi|^{2n}g_0$.

The non-Abelian connection of Eq.~\eqref{eq:Ahat} now becomes
\beq
\hat{A} = -\Gamma_{2n}t_0+\frac\i2\sqrt{n}\left(e_{2n}t_{-}-\bar{e}_{2n}t_{+}\right),
\eeq
whose pullback by $f$ reads
\beq
f^*\hat{A} = -\left(nA+\Gamma_0\right)t_0+\frac{\i\sqrt{n}}{2}\left(\phi^{n}e_0t_{-}-\bar{\phi}^{n}\bar{e}_0t_{+}\right).
\eeq
Computing the curvature, we now have
\begin{align}
  F_{f^{*}\hat{A}}
  &=-n\left[\d A-\left(-C_0+C_{2n}|\phi|^{2n}\right)\frac{\i}{2}e_0\wedge\bar{e}_0\right]t_0\\
&\phantom{=\ }+\frac{\i}{2}n^{\frac32}\left[\phi^{n-1}\left(\d\phi -\i A\phi\right)\wedge e_0 t_{-}-\bar{\phi}^{n-1}\left(\d\bar{\phi} +\i A\bar{\phi}\right)\wedge \bar{e}_0 t_{+}\right],
\end{align}
where the commutation relations remain those of Eq.~\eqref{eq:commutators} and we have used the structure equation for the complex co-frame \eqref{eq:local_frame_new}.
The vanishing of $F_{f^*\hat{A}}$ leads to all the vortex equations, but with the $n$-dependence removed from the second vortex equation, as in Eq.~\eqref{eq:vtx2_2n_Gudnason:2022_repeat}.

For the story upstairs, following the diagram in Sec.~\ref{sec:cartan_geometry}, the Maurer-Cartan structure remains the same, but the projection and section will now acquire the following $n$-dependence
\begin{align}
\pi^*e &= -\frac{\i}{\sqrt{n}}\frac{\bar{z}_{1}}{z_{1}}\sigma,\qquad&
s^*\sigma &= \i\sqrt{n} e,\\
\pi^*\Gamma&=-\sigma^{0}+\i\d\ln\left(\frac{z_{1}}{\bar{z}_{1}}\right),\qquad&
s^*\sigma^0 &= -\Gamma.
\end{align}
Using these everything from Secs.~\ref{sec:cartan_geometry} and \ref{sec:zeromodes} then follows through again, but now we are in exactly the setting of \cite{Ross:2021} with no extra factors of $n$ needed upstairs.

\section{Conclusion and outlook}\label{sec:conclusion}

In this paper we have considered the relationship between the new integrable vortex equations of Ref.~\cite{Gudnason:2022} and the geometric interpretation of vortex equations from Ref.~\cite{Ross:2021}. This provides a geometric description of all of the new vortex equations. 
While the vortex polynomial does not make sense for a non-integer value of $n$, the calculations carried out here, and the geometric interpretation do not require $n$ to be an integer. Thus all of these results for the infinite family of vortex equations generalise to the case where $n\in\mathbb{R}_{>0}$, i.e.~a real positive number.

We have only explored the construction of vortex zero-modes when $C_{0}\neq 0$. However, zero-modes, sometimes known as harmonic spinors, can be constructed in this case by centrally extending $\HH{0}$ to the Nappi-Witten group. This is carried out in Ref.~\cite{Ross:2026}.

There are a few other future directions. The geometric interpretation of \cite{Ross:2021} relies on the uniformisation theorem that all connected Riemann surfaces are equivalent to either $\C,S^{2},\mathbb{H}^{2}$, or quotients of these by a discrete group. The analogue of this for vortex configurations on group manifolds is Thurston's geometrisation conjecture \cite{thurston1997three}, which has eight model geometries. Three of these model geometries have been used so far, but it is not yet known if vortex configurations can be constructed on the others.

\subsection*{Acknowledgements}
S.~B.~G.~thanks the Outstanding Talent Program of Henan University for
partial support.
During the final stage of work on this project C.~R.~acknowledges funding from the Edge Hill Research Investment Fund.

\bibliographystyle{unsrt}
\bibliography{vortices.bib}

\begin{thebibliography}{10}

\bibitem{Manton2}
Nicholas~S. Manton.
\newblock {Five Vortex Equations}.
\newblock {\em J. Phys. A}, 50(12):125403, 2017.

\bibitem{Bradlow:1990ir}
S.~B. Bradlow.
\newblock {Vortices in holomorphic line bundles over closed Kahler manifolds}.
\newblock {\em Commun. Math. Phys.}, 135:1--17, 1990.

\bibitem{Gudnason:2017jsn}
Sven~Bjarke Gudnason and Muneto Nitta.
\newblock {Some exact Bradlow vortex solutions}.
\newblock {\em JHEP}, 05:039, 2017.

\bibitem{CD}
Felipe Contatto and Maciej Dunajski.
\newblock {Manton{\textquoteright}s five vortex equations from self-duality}.
\newblock {\em J. Phys. A}, 50(37):375201, 2017.

\bibitem{Ross:2021}
Calum Ross.
\newblock {Cartan connections and integrable vortex equations}.
\newblock {\em J. Geom. Phys.}, 179:104613, 2022.

\bibitem{RS1}
Calum Ross and Bernd~J Schroers.
\newblock {Magnetic Zero-Modes, Vortices and Cartan Geometry}.
\newblock {\em Lett. Math. Phys.}, 108(4):949--983, 2018.

\bibitem{RS2}
Calum Ross and Bernd~J. Schroers.
\newblock {Hyperbolic vortices and Dirac fields in 2+1 dimensions}.
\newblock {\em J. Phys. A}, 51(29):295202, 2018.

\bibitem{Gudnason:2022}
Sven~Bjarke Gudnason.
\newblock {Nineteen vortex equations and integrability}.
\newblock {\em J. Phys. A}, 55(40):405401, 2022.

\bibitem{Bobenko:2017zso}
Alexander~I. Bobenko and Ananth Sridhar.
\newblock {Abelian Higgs Vortices and Discrete Conformal Maps}.
\newblock {\em Lett. Math. Phys.}, 108(2):249--260, 2018.

\bibitem{Gudnason:2021bkw}
Sven~Bjarke Gudnason and Calum Ross.
\newblock {Magnetic impurities, integrable vortices and the Toda equation}.
\newblock {\em Lett. Math. Phys.}, 111(4):100, 2021.

\bibitem{Manton1}
N.~S. {Manton}.
\newblock {Vortex solutions of the Popov equations}.
\newblock {\em Journal of Physics A Mathematical General}, 46(14):145402, April
  2013.

\bibitem{Witten1}
Edward Witten.
\newblock Some exact multipseudoparticle solutions of classical yang-mills
  theory.
\newblock {\em Phys. Rev. Lett.}, 38:121--124, Jan 1977.

\bibitem{Ross:2026}
Calum Ross and Ra{\'u}l S{\'a}nchez~Gal{\'a}n.
\newblock {Vortex Harmonic Spinors on the Nappi–Witten Space}.
\newblock In preparation.

\bibitem{thurston1997three}
W.P. Thurston.
\newblock {\em Three-dimensional Geometry and Topology}.
\newblock Number~1. Princeton University Press, Princeton, 1997.

\end{thebibliography}

\end{document}